\documentclass[article,12pt]{article}
\usepackage{amsmath,amssymb,cite}
\usepackage{amsthm}
\usepackage{graphics,epsfig,color}
\usepackage[mathscr]{euscript}
\usepackage{thmtools}
\usepackage{dsfont}
\usepackage{xfrac}

\usepackage{comment}
\usepackage{enumerate}

\usepackage{graphicx}
\usepackage{graphics}
\usepackage{caption}
\usepackage{subcaption}
\usepackage{hyperref}
\usepackage{tikz}
\usetikzlibrary{arrows}

\usepackage{tabularx}
\usepackage{bm}
\usepackage{multirow}
\usepackage{pifont}

\usepackage{ bbold }

\theoremstyle{plain}
\newtheorem{theorem}{Theorem}[section]
\newtheorem{proposition}[theorem]{Proposition}

\newtheorem{corollary}[theorem]{Corollary}

\theoremstyle{definition}
\newtheorem{definition}[theorem]{Definition}

\theoremstyle{remark}

\newtheorem{example}[theorem]{Example}

\numberwithin{equation}{section}
\numberwithin{theorem}{section}

\newcommand{\wc}{\mathcal{W}}
\newcommand{\wcm}{\mathcal{W}_{\min}}
\renewcommand{\epsilon}{\varepsilon}

\renewcommand{\hat}{\widehat}

\newcommand{\iind}{\mathcal{I}}

\title{A Coopetition Index for Coalitions in Simple Games}

\author{Michele Aleandri\footnote{ORCID:0000-0002-5177-8176, LUISS University, Viale Romania 32, 00197 Rome, Italy. \textit{maleandri@luiss.it}. Member of the Gruppo Nazionale per l’Analisi Matematica,  la Probabilità e le loro Applicazioni (GNAMPA) of the Istituto Nazionale di Alta Matematica (INdAM) and is supported by Project of Significant National Interest – PRIN 2022 of title “Impact of the Human Activities on the Environment and Economic Decision Making in a Heterogeneous Setting: Mathematical Models and Policy Implications”- Codice Cineca: 20223PNJ8K- CUP I53D23004320008\\
Conflict of interests: Michele Aleandri declares that he has no conflict of interest.}  \and Marco Dall'Aglio\footnote{ORCID: 0000-0002-2243-1026, LUISS University, Viale Romania 32, 00197 Rome, Italy. \textit{mdallaglio@luiss.it}. Member of the Gruppo Nazionale per l’Analisi Matematica,  la Probabilità e le loro Applicazioni (GNAMPA) of the Istituto Nazionale di Alta Matematica (INdAM).\\
Conflict of interests: Marco Dall'Aglio declares that he has no conflict of interest. }}   

\begin{document}

	\maketitle
	\begin{abstract}

In simple games, larger coalitions typically wield more power, but do all players align their efforts effectively? Consider a voting scenario where a coalition forms, but needs more voters to pass a bill. The cohesion of the new group of voters hinges on whether all the  new members can proficiently collaborate with the existing players to ensure the bill's passage or if subgroups form that pursue an independent alternative, thus generating  antagonism among the new voters.

This research introduces two classes of  \textit{coopetition indices} -- one relative and one absolute, the latter ranging from -1 to 1, to measure agents' preferences for cooperation (when positive) or competition (when negative) with the remaining players. These indices, together with a generalized group value, provide a comprehensive picture of the relevance and the cohesion of groups. We discuss the relationship with similar group indices and provide proper coopetition Banzhaf and Shapley-Owen types of indices.

By applying our indices to the apex game and symmetric majority games, we observe that cooperation and competition frequently balance each other out, leading to null values for the Shapley-Owen and Banzhaf coopetition indices. An electoral application with real world data is also considered.

	\end{abstract} 
 
\section{Introduction}

In monotone simple games, larger coalitions generally possess greater power. However, not all coalitions can guarantee cohesion equally. Consider a voting body where a coalition of voters forms with the intent of passing a bill but needs to recruit additional voters to achieve this goal. If every voter in the new coalition is essential to pass the bill, the incoming group might exhibit total cooperation. Conversely, if disjoint subgroups within the new coalition suffice to pass the bill, antagonism might arise among the new members.

Cooperative game theory addresses the interactions among multiple agents, with much focus on indices that evaluate the importance of individual agents, such as the Shapley value. However, recent research has shifted towards assessing the power of coalitions as groups, using generalized group values?.

This research aims to define an index that measures the attitude of a group of players in terms of reciprocal cooperation or competition. Several measures to assess the benefit of coalition members working together have already been proposed.
We refer, first of all, to the profitability index by Derks and Tijs \cite{derks2000merge}, which compares the power of a coalition measured by generalized power indices against the total power of individuals within the coalition. Despite its theoretical soundness, this index has paradoxical outcomes, like encouraging alliances that may reveal ineffective if its players do not merge as a single unit (see Example \ref{ex:profitparadox} later).

We are aware of two other existing approaches, those by Grabisch and Roubens \cite{grabisch1999axiomatic} and by Hausken and Mohr \cite{hausken2001value} that, when applied to pairs of agents, are capable of distinguishing couples that cooperate because they are both essential for a winning outcome and those who engage in a direct competition because the individuals are critical on their own. These definitions are designed with aims which are different from those that we are pursuing in this work and cannot be used to measure the degree of collaboration for groups of players larger than size 2.

Motivated by these findings, we propose a new class of \textit{coopetition indices} that indicate whether the players of a given coalition prefer cooperation (when positive) or competition (when negative). Such measures are bounded above and below, respectively, by the generalized group value of the coalition defined by the same probabilistic setting and by the negative of the same value. Together, the coopetition and the (generalized) group value provide a synthetic indicator of the strength and the cohesion of a group of players. Additionally, but no less importantly, the ratio of the two measures for non-null coalitions provide an absolute index of coopetition that ranges between -1 and 1, no matter what size or power is the coalition.

We introduce the Banzhaf  and the Shapley-Owen coopetition indices inspired by the classical indices for individuals. We then highlight some properties of the indices, focusing at first on null players: A coalition of null players has null coopetition index, reflecting neutrality between the two opposing attitudes, and the inclusion of a null player to any coalition will diminish the overall attitude of a coalition towards one of the two extremes of the range bringing the coopetition index closer to zero.

The second property affects those sets of players that contain a minimal winning coalition: replacing that minimal winning coalition with a couple of minimal winning coalitions that cover the same players leads to a more fractured environment and, therefore to a higher degree of competition which lowers the coopetition index.

We apply our new index to the apex game and to the symmetric majority game, settings where cooperation and competition naturally coexist and often balance, resulting in null Shapley-Owen and Banzhaf coopetition indices. We then consider an application to the 2015 Spanish general elections already examined in Flores et al.\cite{flores2019shapley}.

The paper is organized as follows. Section \ref{sec:motivation} introduces the notation and motivates the work within the existing literature. In Section \ref{sec:newindex} the two new classes of indices that depend on the choice of families of probability distributions are defined. Section \ref{sec:BS} delves into specifications of the coopetition indices according to the Banzhaf and the Shapley-Owen patterns. Section \ref{sec:property} verifies some properties of the new indices. Section \ref{sec:examples} applies the new indices to a couple of examples that are often taken as reference: the apex game and the simple majority game. In Section \ref{sec:elect} an application based on actual electoral data is studied.  Section \ref{sec:conclusion} presents our conclusions, followed by an Appendix which gathers some proofs.

\section{Motivation of the work}
\label{sec:motivation}
\subsection{Basic Notation}\label{sec:notation}

We denote the cardinality of a finite set with the corresponding small letter: e.g. \(n=|N|\), \(s=|S| \), \(t=|T|\). Given a set $A\subseteq N$, its complement is defined as $A^c=N\setminus A$. Also, 
we will often omit braces for singletons, e.g. writing $v(i), S \backslash i$ instead of $v(\{i\}), S \backslash\{i\}$. Similarly, we will write $i j, ijk$ instead of $\{i, j\}, \{i,j,k\}$. 

A (monotonic) \emph{simple game} is a pair $(N,v)$, where $N=\{1,2,\ldots,n\}$ denotes the finite set of players and $v:2^N\to\{0,1\}$ is the \emph{characteristic function}, with $v(\emptyset)=0$, $v(S)\leq v(T)$ for all $S,T$ subsets of  $N$ such that $S\subseteq T$, and $v(N)=1$. Unless otherwise specified, we always take the simple game $(N,v)$.\\
	Given a coalition $S \subseteq N$, if $v(S)=0$, then $S$ is a \emph{losing} coalition, while if $v(S)=1$, then $S$ is a \emph{winning} coalition.   We let $\wc^v=\{S \subseteq N: v(S)=1\}$ be the set of winning coalitions and  let $\wcm^v=\{W\in\wc^v: \nexists S\in\wc^v, S\subset W\}$ be the set of minimal winning coalitions.
    
 The marginal contribution of coalition $S$ to coalition $T$, $T\subseteq N\setminus  S$ is defined as:
	\[
	v_S'(T)= v(S \cup T )-v(T ). 
	\]
	A coalition $S$ is critical with respect to (wrt) a losing coalition $T$ if $v_S'(T)=1$; in particular, a player $i$ is critical wrt $T$ if $v_{i}'(T)=1$. A player $i$ such that $v_i'(T)=0$ for all $T\subseteq N\setminus\{i\}$ is called \emph{null player}.

\subsection{Towards a New Index}
We now review existing group indices and their interpretations, highlighting that no current measure fully answers the following research question: \emph{How can we measure the internal attitude of a group of agents towards cooperation or competition when forming a coalition?} In other words, when agents decide to form a coalition, will they exhibit peaceful coexistence and collaborate effectively, or will their interactions be marked by intense internal rivalry - or perhaps a balance between these forces?

Before presenting our own definition, we briefly recount some available group indices that serve as important reference points for our analysis. For clarity, we adopt the context originally used by Shapley for the eponymous index, although many of these measures admit straightforward generalizations that account for the random formation of alternative coalitions.

In measuring the variation in agents' power when they work in groups rather than individually, Marichal and coauthors \cite{marichal2007axiomatic} introduce the concept of the Shapley generalized value. The following definition formalizes this concept.

\begin{definition}[Marichal et al.\ \cite{marichal2007axiomatic}]\label{def:GSV}
For any coalition \( S \subseteq N \), the \emph{Shapley generalized value} is defined as
\[
\Phi_{Sh}^v(S) = \sum_{T \subseteq N \setminus S} \frac{(n-s-t)!\,t!}{(n-s+1)!}\,v'_S(T).
\]
When \( S \) is a singleton, this definition coincides with the classical Shapley value.
\end{definition}

This index measures the strength of a group when its unity is so strong that the individual players effectively merge into a single entity. Flores et al.\ \cite{flores2019evaluating} recently provided an alternative axiomatization to the original one, together with several applications. 

Comparing the group's performance with the sum of individual outcomes provides a notion of the profitability of merging.
Building on this idea, Derks and Tijs \cite{derks2000merge} propose the following measure.

\begin{definition}[Derks and Tijs \cite{derks2000merge}]
For any coalition \( S \subseteq N \), the \emph{profitability index} is defined as
\[
P^v(S)=\Phi_{Sh}^v(S) - \sum_{i\in S}\Phi_{Sh}^v(i).
\]
\end{definition}

This index, which ranges between \(-1\) and \(1\), is positive when it is beneficial for the agents to merge. Note that the concept of profitability is grounded in the efficiency of the Shapley value; the gains of a merging group balance the losses incurred by the remaining players. As a consequence, extending this notion to other probabilistic values is problematic. Moreover, the profitability index effectively compares two different games: one in which a group of players merges into a single entity and another in which all players act independently. This shift in perspective can lead to paradoxical conclusions.

\begin{example}
\label{ex:profitparadox}
Consider \(N=\{1,2,3,4,5\}\) and a simple game \((N,v)\) with set of minimal winning coalitions given by \(\mathcal{W}^v_{\min} = \{ 123, 145, 35  \} \). Players 3 and 4 may consider to join their forces since \(P^v(34) = 1/6 > 0\). An explanation of this advantage can be seen by considering that, when players 3 and 4 merge and become a single player \(\mathbf{p}\), the new game \((N\setminus 34\cup\mathbf{p}, v')\) has minimal winning coalition set given by \(\mathcal{W}^{v'} = \{ 12\mathbf{p}, \mathbf{p}5  \} \) and \(\mathbf{p}\) gains importance at the disadvantage of the other players.  Yet, upon closer examination, in the original game \(v\) one finds that players 3 and 4 are never jointly critical: whenever player 3 is critical, player 4 may or may not be so, but player 4 is never critical without player 3 being so. Thus, there is no intrinsic incentive for player 3 to merge with player 4 unless the merger occurs before the game is played, and this new entity will be evaluated only wrt the coalitions that do not contain any of the two players.

\end{example}

This observation naturally leads to the question of whether one can devise an index that directly measures the willingness of coalition members to cooperate. 

We take a different approach and consider the following refinements of the notion of criticality that indicate two opposite attitudes of the players within a coalition.
\begin{definition}
    A  player $i\in S$ is \emph{essential} for $S$ being critical wrt $T$ if $S\setminus i$ is not critical wrt $T$. A coalition $S$ is called \emph{essential critical}, or simply \emph{essential},  wrt a coalition $T$, if each agent $i\in S$ is essential for $S$ being critical wrt $T$. Two disjoint subcoalitions \(S_1,S_2 \subset S\) are called \emph{complementary critical} wrt $T$ if both \(S_1\) and \(S_2\) are critical for \(T\). A coalition $S$ is called \emph{totally complementary critical} wrt $T$ if each player in $S$ is critical wrt $T$.  
\end{definition}
For a more comprehensive introduction to these refinements of the notion of essential criticality we refer to Aleandri and Dall'Aglio \cite{ALEANDRI202513}. 

These notions specify whether the players of a coalition are willing to work together because they are essential to the successful outcome of a coalition and they are all needed in the effort, or they will face  a stark competition because two groups of players suffice on their own to make a coalition win. Our aim is to define an index to reveal which of the attitudes just outlined will prevail within a coalition.
Restricting our attention to pairs of agents, the following simple indicator, defined  for each \(i,j \in N\) and $T\subseteq N\setminus ij$, is capable of signaling the behavior of each couple in terms of the previous definitions:
\begin{equation}
\label{eq:attcouples}
c^v(ij,T) = v(ij \cup T) - v(i \cup T) - v(j \cup T) + v(T).
\end{equation}
In fact, we have
\begin{itemize}
    \item $c^v(ij, T)=1$ if coalition $ij$ is \emph{essential critical} for $T$;
    \item $c^v(ij, T)=-1$ if \(ij\) is \emph{totally complementary critical} for $T$;
    \item $c^v(ij, T)=0$ otherwise.
\end{itemize}
Note that the latter case puts together the cases when the couple \(ij\) is critical due to the action of only one of the two players and that in which the couple as a whole is not critical wrt \(T\).

We are aware of two indices that consider averages of \eqref{eq:attcouples} over all possible choices of coalitions outside the considered pair, but they were defined with different aims. The first 
was proposed by Grabisch and Roubens \cite{grabisch1999axiomatic} for transferable utility (TU) games, with the aim of quantifying the degree of interaction among players in a coalition \(S\) of any size. The Shapley interaction index for any coalition \( S\subseteq N \) is defined as follows.

\begin{definition}[Grabisch and Roubens \cite{grabisch1999axiomatic}]
For any coalition \( S\subseteq N \), the \emph{Shapley interaction index} is given by
\[
\iind^v_{Sh}(S)=\sum_{T \subseteq N \setminus S} \frac{(n-t-s)! \, t!}{(n-s+1)!} I^v(S,T),
\]
where
\[
I^v(S,T)=\sum_{L \subseteq S} (-1)^{s-l} v(L \cup T).
\]
\end{definition}
The index satisfies four classical axioms: symmetry, linearity, the irrelevance of dummy players, together with one of two equivalent recursive relationship linking the interaction indices of a coalition and its subsets. Clearly, \(I^v(ij,T) = c^v(ij,T)\), and the interaction index for couples suits our needs. 

Another proposal is the following.

\begin{definition}
    (Hausken and Mohr \cite{hausken2001value}) The value of player \(i \in N\) to player \(j \in N\) is given by
    \begin{equation}
        \label{def:hauskenvalue}
         \eta^{v}(ij)=\sum_{T \subseteq N \setminus ij} \frac{(t+1)!(n-t-2)!}{n!} \;  c^v(ij,T) \sum_{u=t+2}^{n} \frac{1}{u}.
    \end{equation}
\end{definition}

The aim of this proposal is to provide a decomposition of the Shapley value of an individual  into \(n\) components, each one revealing the value of a player to any other player (including themselves), meant as the player's expected contribution to every other player's value. Although this index has a different scope, it is a valid alternative to the interaction index in measuring the attitude of the pair of players.

\begin{example}[continues=ex:profitparadox]
In the game described above, we can show that \( \iind^v_{Sh}(34)=-\frac{1}{3} \), and \( \eta^v(34)=-\frac{1}{16}\). Both indices  indicate a conflicting attitude between players 3 and 4 which would compromise the stability of their union.
\end{example}

It might be reasonable to assume that these indices could be adopted to measure the internal degree of cooperation or competition in a group.
One might therefore expect the interaction index to return a value close to the group Shapley value when a coalition is effective only if all players act together, and a value near the negative of the group Shapley value when agents can perform effectively on their own. Yet, as shown in the following example, the interaction index oscillates between negative and positive values as the size of the minimal (or essential) critical coalitions within the triplet grows.

\begin{example}
\label{ex:interparadox}
Consider \( N=\{1,2,3,4\} \) and three games \( v \), \( w \), and \( u \) with the sets of minimal winning coalitions given respectively by
\[
\wcm^v=\{123\},\quad \wcm^w=\{12,23,13\},\quad \text{and} \quad \wcm^u=\{1,2,3\}.
\]
A straightforward computation yields
\[
\iind_{Sh}^v(123)=\iind_{Sh}^u(123)=1,\quad \text{while} \quad \iind_{Sh}^w(123)=-2.
\]
\end{example}
In fact, for simple games, the quantity \(I^v(ijh,T)\) may vary between \(-2\) and \(1\), being negative when critical subcoalitions of size 2 prevail, and positive when critical subcoalitions of size 1 or 3 dominate. This is coherent with the axioms that characterize such index: symmetry, linearity, the irrelevance of dummy players, and a recursive relationship linking the interaction indices of a coalition and its subsets. One of these axioms describes the index as the difference in the interaction of \(i\) and \(j\) with or without \(h\) (but the roles of the players are interchangeable).
\begin{example}[continues=ex:interparadox]
It can be easily verified that in the games \(v\) and \(u\) the presence of one of the first three players increases the interaction between the other two by one unit, while in \(w\) it decreases the same difference by two units.
\end{example}

The value of a player to another player has been recently extended by Hausken \cite{hausken2020shapley} to replace the players with coalitions, not necessarily disjoint, as the sum of the values of one player in the first coalition to a player in the second one.
It is not clear, however, how this index could be adapted to measure the internal cohesion of a coalition.  We note that this index, being  a sum of the corresponding index for pairs of individuals, will not provide new information on the association of the players.

This lack of measures for the cooperative/competitive attitude for coalitions of size larger than 2 motivates the search for a new index for coalitions of size 2 or larger that takes its maximum value when every time the coalition is critical, it is essentially critical (all its members are needed to change the outcome) and its minimum value when it is totally complementary critical. The examples of this section show that no existing candidate succeeds in this task, simply because they are designed for different tasks. We will define a new index, denoted {\em coopetition index} that satisfies the following properties: (i) regardless of coalition size, its value should lie between the group Shapley value and its negative, with the extremes representing maximum cooperation and competition,  (ii) the index should be zero for the null coalition, reflecting neutrality; (iii) the inclusion of a null player to any coalition should drive the index for that coalition toward zero, as their presence dilutes the coalition's cooperative or competitive strength; and (iv) if a minimal winning coalition is replaced by two smaller minimal winning coalitions involving the same players, the coopetition index should decrease for those coalitions and their supersets, indicating a heightened degree of internal competition.     
Property (i) enables us to consider another index as the ratio between the coopetition index and the group Shapley value. This new measure lies between 1 and \(-1\), independent of coalition size and will be denoted as the {\em absolute coopetition index}.

\section{Balancing Cooperation and Competition in a Single Index}\label{sec:newindex}
 We start from formula \eqref{eq:attcouples} and provide an extensions to larger coalitions. Take a coalition $S\subseteq N$, with $|S|\geq 2$,  and call $\Pi_2(S)$ the set of non-trivial (i.e.\ non-empty) $2$-partitions, i.e.  
$$\Pi_2(S)=\big\{ \pi=\{S_1,S_2\}:\,\emptyset \neq S_1,S_2\subset S,\, S_1\cup S_2 = S \mbox{ and } S_1\cap S_2 =\emptyset  \big\}.$$

\begin{definition}
Take $S\subseteq N$. For any  2-partition $\pi=\{S_1,S_2\}\in\Pi_2(S)$ the \emph{block interaction indicator} between the sets $S_1$ and $S_2$, against $T\subseteq  N\setminus S$, is defined as
\begin{equation*}
    BI^v(\pi,T)=v(S_1 \cup S_2 \cup  T) - v( S_1 \cup  T) - v( S_2 \cup  T) + v(T).
\end{equation*}

Take the family of probability distribution  $p=\{p_S\}_{S\subseteq N}$, where each $p_S$ is a probability distributions over the set $\Pi_2(S)$. The \emph{attitude} of coalition $S$ towards $T$ is defined as
\begin{equation*}
    \mathcal{A}^v_{p}(S,T) = \sum_{\pi\in \Pi_2( S)} p_S(\pi)BI^v(\pi,  T).
\end{equation*}
\end{definition}
A null value for the attitude of a coalition \(S\) denotes a situation where it is equally likely, according to \(p_S\), to have a partition yielding block interactions of +1 and -1.

In what follows, \(p\) will always indicate a family of distributions \(\{p_S\}_{S\subseteq N}\), with each \(p_S\) defined on \(\Pi_2(S)\).

Note that when $S=ij$, $i,j\in N$, then \(\Pi_2(S)\) reduces to a single element and, therefore, $ \mathcal{A}^v_{p}(ij,T) =BI^v(ij,T)=c^v(ij,T)$ no matter what family \(p\) is chosen.

Our first result shows that the attitude of a coalition towards an other mimics the behaviour of the interaction index for couples, singling out essential critical  and complementary coalitions with more than two players.

\begin{proposition} \label{prop:-1+1}
    Fix \(p\) and let  $T\subset N$ be a losing coalition with $|T|\leq n-2$. For each coalition $S\subseteq N\setminus T$, $|S|\geq 2$, critical wrt $T$ and for all  strictly positive  probability distributions $p_S$ over the set $\Pi_2(S)$, i.e. $p_S(\pi)>0$ for all $\pi\in\Pi_2(S)$, then $i)$ $ \mathcal{A}^v_{p}(S,T)=1$ if and only if  $S$ is essential critical wrt $T$ and $ ii)$ $ \mathcal{A}^v_{p}(S,T)=-1$ if and only if $S$ is complementary critical wrt $T$.
\end{proposition}
\begin{proof}
The "if" part is easy to obtain and left to the reader. By definition $BI^v(\pi, T)$ can only assume value "1" or "-1", then, if $\mathcal{A}^v_{p}(S,T)=1$, it holds $BI^v(\pi, T)=1$ for all $\pi\in\Pi_2(S)$. This implies that $S$ is essential critical wrt $T$.  The case $\mathcal{A}^v_{p}(S,T)=-1$ is similar. 
\end{proof}

We are now ready to define the indices that measure the cooperation and antagonism of players to form a coalition. 
\begin{definition}
    \label{def:gCoopIndex}
Take two families of  probability distribution $p$ and $q=\{q_{S}\}_{S\subseteq N}$, where each $q_S$ is a probability distributions over $2^{N\setminus S}$. The \emph{coopetition index } for coalition $S\subseteq N$ is defined as
\begin{equation*}
    \mathcal{C}^v_{p,q}(S) = \sum_{T\subseteq N\setminus S} q_S(T)\mathcal{A}^v_{p}(S,T).
\end{equation*}
\end{definition}

In what follows, \( q \) will always denote a family of probability distribution $q=\{q_{S}\}_{S\subseteq N}$, with each \(q_S\) defined on \(2^{N \setminus S} \).

We recall the definitions of the probabilistic generalized value  and generalize semivalue given by Marichal et al. \cite{marichal2007axiomatic}. 
\begin{definition}
  Given a family of probability distributions \(q\), a \emph{probabilistic generalized value} of a coalition $S \subseteq N$ is defined by
$$
\Phi_q^v(S)=\sum_{T \subseteq N \backslash S} q_S(T)v_S'(T).
$$
\\
A \emph{generalized semivalue} is a probabilistic generalized value such that, additionally, for any $S \subseteq N$, the probability $q_S(T)$, $\forall T\subseteq N\setminus S$, depends only on the cardinalities of the coalitions $S, T$, and $N$, i.e., for any $s \in\{0, \ldots, n\}$, there exists a family of nonnegative real numbers $\left\{q_s^n(t)\right\}_{t=0, \ldots, n-s}$ fulfilling
$$
\sum_{t=0}^{n-s}\binom{n-s}{t} q_s^n(t)=1
$$
such that, for any $S \subseteq N$ and any $T \subseteq N \backslash S$, we have $q_S(T)=q_s^n(t)$.
\end{definition}

Notable examples of generalized semivalues are the Shapley Generalized Value of Definition \ref{def:GSV}, obtained by setting \(q_s^n(t)=\frac{(n-s-t)!t!}{(n-s-1)!}\) and the Banzhaf Generalized value when \(q_s^n(t)=2^{s-n}\).

The block interaction indicator ranges between \(-1\) and \(1\) and the same bounds apply to the coopetition index. The next result shows that sharper bounds hold. 
\begin{proposition}
\label{prop:sharpbounds}
    For any pair of  families of probability distributions \(q\) and \(p\), and for every \(S \subseteq N\), with $|S|\geq 2$, it holds
    \begin{equation*}
        -\Phi^v_q(S) \leq \mathcal{C}^v_{p,q}(S) \leq \Phi^v_q(S) \; .
    \end{equation*}
\end{proposition}
\begin{proof}
  Both inequalities follow from a rearrangement of the definition of the block interaction indicator. For the second inequality we have
  \begin{equation}
      \label{eq:blockrearrangement}
BI^v(\pi,T)= v(S\cup T)-v(T) -\big(v(S_1\cup T)-v(T)+v(S_2\cup T)-v(T)\big).
  \end{equation}
Now, the part in brackets is nonnegative, so \( BI^v(\pi,T)\leq v(S\cup T)-v(T)\), and, therefore,
\[
  \mathcal{C}^v_{p,q}(S) \leq \sum_{T\subseteq N\setminus S} q_S(T)\big(v(S\cup T)-v(T)\big)\sum_{\pi\in \Pi_2( S)} p_S(\pi)=\Phi_q^v(S).
\]
A different rearrangement of the block interaction indicator yields the first inequality, since
  \begin{multline*}
      BI^v(\pi,T)=-v(S\cup T)+v(T) +\big(v(S\cup T)-v(S_1\cup T)+v(S\cup T)-v(S_2\cup T)\big) \\
      \geq -v(S\cup T)+v(T),
  \end{multline*}
  for all $S\subseteq N$, $\pi\in\Pi_2(S)$, and $T\subseteq N\setminus S$. Then 
  \[\mathcal{C}^v_{p,q}(S) \geq \sum_{T\subseteq N\setminus S} q_S(T)\big(-v(S\cup T)+v(T)\big)\sum_{\pi\in \Pi_2( S)} p_S(\pi)=-\Phi_q^v(S).
  \]
\end{proof}

The following result shows that the new bounds are attainable.
\begin{corollary}
\label{cor:attain}
For any given pair of families of probability distributions \(q\) and \(p\), the followings hold:
    \begin{enumerate}[i)]
        \item If every time \(S\) is critical  for some coalition \(T \subseteq N \setminus S\) it is essential critical for the same coalition, then \(\mathcal{C}_{pq}^v(S) = \Phi_q^v(S)\).
        \item If every time \(S\) is critical  for some coalition \(T \subseteq N \setminus S\) it is complementary critical for the same coalition, then \(\mathcal{C}_{pq}^v(S) = - \Phi_q^v(S)\).
    \end{enumerate}
\end{corollary}

Measuring both the generalized group value and the coopetition index enables us to give an informative picture of the coalition strength and cohesion. Additionally, the previous result also shows that a ratio between the two indices provides an absolute coopetition index which has a common range for any coalition.

\begin{definition}
Take any pair of families of probability distributions \(q\) and \(p\). The \emph{absolute coopetition index } for coalition $S\subseteq N$ is defined as
\begin{equation*}
    \hat{\mathcal{C}}^v_{p,q}(S) = \begin{cases}
        \cfrac{\mathcal{C}^v_{p,q}(S)}{\Phi^v_q(S)}& \mbox{ if } \Phi^v_q(S)\neq 0,\\
        0 & \mbox{otherwise},
    \end{cases}
\end{equation*}
where \(\Phi^v_q(S)\) is the probabilistic generalized value of \(S\).
\end{definition}

Clearly, Proposition \ref{prop:sharpbounds} implies that, for every \(S \subseteq N\),
\[
-1 \leq \hat{\mathcal{C}}^v_{p,q}(S) \leq 1,
\]
and Corollary \ref{cor:attain} provides the conditions under which the bounds are attained.\\

The following result provides a  useful characterization of the coopetition index.

\begin{proposition}
Take a pair of  families of probability distributions $q$ and $p$. Suppose that for fixed \(S \subseteq N\) there exist another family of probability distributions $q^*=\{q^*_{R}\}_{\emptyset \neq R \subsetneq S}$  each one defined over the set $2^{N\setminus S}$, such that 
\begin{equation}\label{eq:charac} q_S(T)=q_S^*(T),\qquad \forall T\subseteq N\setminus S.
\end{equation}
Then
\begin{equation}
    \label{eq:newchar}
    \mathcal{C}^v_{p,q}(S)=\Phi^v_q(S)-\sum_{\pi=\{S_1,S_2\} \in \Pi_2(S)} p_S(\pi)\left( \Phi^{-S_2}_{q^*}(S_1) + \Phi^{-S_1}_{q^*}(S_2)  \right),
\end{equation}
where \(\Phi^v_q(S)\) is the probabilistic generalized value of \(S\) under distribution \(q\), and, for \(\emptyset \neq R \subsetneq S\), \(\Phi^{-R}_{q^*}(S \setminus R)\) is the  probabilistic generalized value of \(S \setminus R\) under distribution \(q^*\) for the game \( (N\setminus S,v) \).
\end{proposition}
\begin{proof} 
For given families \(p\) and \(q\), and \(S \subseteq N\), we can use again \eqref{eq:blockrearrangement} to obtain
\begin{align*}
    \mathcal{C}_{p,q}^v(S) &= \sum_{T \subseteq N \setminus S} q_S(T) \sum_{\{S_1,S_2\} \in \Pi_2(S)} p_S(\{S_1,S_2\}) BI^v \left( \{S_1,S_2\},T \right)\\
    &=   \sum_{T \subseteq N \setminus S} q_S(T) \left( v(S \cup T) - v(T) \right) \\
    &\quad-\sum_{\{S_1,S_2\} \in \Pi_2(S)} p_S(\{S_1,S_2\}) \sum_{T \subseteq N \setminus S} q_S(T) \left(v(S_1 \cup T) - v(T) + v(S_2 \cup T) - v(T) \right) \\
    &=\sum_{T \subseteq N \setminus S} q^*_S(T) v'_S(T) - \sum_{\{S_1,S_2\} \in \Pi_2(S)} p_S(\{S_1,S_2\}) \left( \sum_{T \subseteq N \setminus S} q^*_{S_1}(T) v'_{S_1}(T) \right.\\
    &\left. \qquad\qquad\qquad+ \sum_{T \subseteq N \setminus S} q^*_{S_2}(T) v'_{S_2}(T)\right),
\end{align*}
which yields the desired result.
\end{proof}

Observe that condition \eqref{eq:charac} is satisfies for $q$ and $q^*$ as in the Shapley  and Banzhaf Generalized values, and more generally when $q$ and $q^*$ depend only on the cardinality of $N\setminus S$ and $T$. This condition is met when \(q\) and \(q^*\) define a semivalue.

\section{Some important cases}\label{sec:BS}
To define specific coopetition indices to work with in the applications that follow we consider the two consolidated models for coalition formation that preside the definition of the Banzhaf and the Shapley values for individuals.

\subsection{The Banzhaf Coopetition Indices}
We assume that both the partition that defines the attitude and the coalition of external players that interacts with the chosen coalition $S$ are chosen uniformly at random.
\begin{definition}
    Take the family of probability distribution $p$ and assume that, $\forall S\subseteq N$, \(p_S(\pi)=\left( 2^{s-1}-1\right)^{-1}\) for every \(\pi \in \Pi_2(S)\). The \emph{Banzhaf attitude index} of \(S\) wrt \(T \) is defined as
    \begin{equation*}
        \mathcal{A}_{Bz}^v(S,T)=\frac{1}{2^{s-1}-1}\sum_{\pi \in \Pi_2(S) }BI^v(\pi,T).
    \end{equation*}
    Since each partition in \(\Pi_2(S)\) is equally likely, a null attitude index of \(S\) wrt \(T\) describes a situation where each 2-partition with a block interaction indicator of 1 is counterbalanced by another 2-partition with opposite block interaction indicator. 

     Taking the family of probability distribution $q$ defined by  \(q_S(T)=2^{s-n}\), for all $S\subseteq N$  and $T\subseteq N\setminus S$, the \emph{Banzhaf coopetition index} is defined as
        \begin{equation*}
            \mathcal{C}^v_{Bz}(S)=2^{s-n} \sum_{T \subseteq N \setminus S} \mathcal{A}_{Bz}^v(S,T)=\frac{1}{2^{n-s}(2^{s-1}-1)} \sum_{T \subseteq N \setminus S} \sum_{\pi \in \Pi_2(S)} BI^v(\pi,T).
        \end{equation*}
    \end{definition}
A coalition \(S\) has null Banzhaf coopetition index whenever, to any \(T \subseteq N \setminus S\) and \(\pi \in \Pi_2(S)\) such that \(BI^v(\pi,T)=1\), we can associate a coalition \(T' \subseteq N \setminus S\) (which may coincide with \(T\)) and \(\pi' \in \Pi_2(S)\) such that  \(BI^v(\pi',T')=-1\).

 The  Banzhaf absolute coopetition index reads as

 \begin{equation*}
     \hat{\mathcal{C}}^v_{Bz}(S)=\frac{1}{(2^{s-1}-1)} \frac{\sum_{T \subseteq N \setminus S} \sum_{\pi \in \Pi_2(S)} BI^v(\pi,T)}{\sum_{T \subseteq N \setminus S} v'_S(T)}.
 \end{equation*}

\subsection{The Shapley-Owen Coopetition Indices}\label{subsection:SO}
The second model considers players joining a coalition one at a time, with players in the chosen coalition $S$ arriving one immediately after the other and with the overall arrival sequence chosen uniformly at random among the sequences that keep the players in \(S\) close to one another in terms of entry sequence. We follow the Owen construction for a value with a priori union built on the Shapley value in which only the players in \(S\) gather together before the game is played.

Therefore, we consider the sequence of players \(\Sigma_S(N)\) in which the players in \(S\) are located one next to the other. For each \(\sigma \in \Sigma_S(N)\), at a random time of entrance chosen uniformly among the available slots, players already in the coalition and those outside of it evaluate the mutual benefit of their interaction through the block interaction indicator. The average wrt the time of entrance will describe the overall attitude of the players of \(S\) wrt the previous players for the given entry sequence. 

\begin{definition}
    The \emph{Shapley-Owen coopetition index} is defined as
    \begin{equation*}
        C^v_{SO}(S) = \frac{1}{|\Sigma_S(N)|} \sum_{\sigma \in \Sigma_S(N)} AS^v(S,\sigma),
    \end{equation*}
    where $AS^v(S,\sigma)$, called the sequential attitude of $S$ toward $\sigma$, is
    \[
    AS^v(S,\sigma) = \frac{1}{s-1} \sum_{r=1}^{s-1} BI^v(\{F_r^{\sigma}(S),L_r^{\sigma}(S)\},Pre^\sigma(S)),
    \]
    where \(F_r^{\sigma}(S)\) (\(L_r^{\sigma}(S)\) resp.) are the first $r$ (last $s-r$, resp.) players of \(S\) that join the coalition in the sequence \(\sigma\) and \(Pre^\sigma(S)\) are the predecessors of $S$ in the sequence \(\sigma\).
\end{definition}

A null coopetition index for the coalition \(S\) denotes a situation where, to any sequence \(\sigma \in \Sigma_S(N)\) and any choice of time \(r\) of entry for which the players in \(S\), who have already entered and those who are still outside gain mutual advantage from cooperation, there corresponds a sequence \(\sigma' \in \Sigma_S(N)\) and a time of entry \(r'\) such that the two groups are in perfect competition. Notice that it may be the case that \(\sigma=\sigma'\), but then \(r \neq r'\).

The newly introduced coopetition index is a special case of the generalized coopetition index, as the next result shows.

\begin{definition}
 For any $S \subseteq N$, $|S|\geq2$, and $T\subseteq N\setminus S$, the \emph{Shapley-Owen attitude} index is defined as
 \begin{equation*}
     \mathcal{A}^v_{SO}(S,T)= \frac{1}{t!s!(n-s-t)!}\sum_{\substack{\sigma\in\Sigma_S(N)\\ Pre^\sigma(S)=T }}AS^v(S,\sigma).
 \end{equation*}
 
 \end{definition}

Note that, for fixed \(S\) and \(T\), if to any \( \pi \in \Pi_2(S)\) that splits \(S\) into two sets of sizes \(r\) and \(s-r\), respectively, and such that \(BI^v(\pi,T)=1\), we can associate another \( \pi' \in \Pi_2(S)\) that splits \(S\) into subsets of the same sizes as those contained in \(\pi\) and such that \(BI^v(\pi',T)=-1\) then \(\mathcal{A}^v_{SO}(S,T)=0\). For the Shapley-Owen coopetition index, however, a null attitude could arise from a differing number of 2-partitions of \(S\) having positive or negative block interaction indicators. This is due to the fact that \(p_S\) in \(p\) assigns different probabilities to the elements in \(\Pi_2(S)\) of differing sizes. 
  
\begin{proposition}
    For any $S \subseteq N$, $s\geq2$, and $T\subseteq N\setminus S$ then,
    \begin{align*}
        \mathcal{A}^v_{SO}(S,T) =\mathcal{A}^v_{p}(S,T),
    \end{align*}
    where $p$ with
    \begin{equation*}
       p_S(\pi)=\frac{2 r!(s-r)! }{(s-1) s!}, \qquad \mbox{for any } \pi=\{R,S\setminus R\} \in \Pi_2(S).
    \end{equation*}
    Moreover, taking $q$ with
    \begin{equation*}
           q_S(T)=\frac{t!(n-s-t)!}{(n-s+1)!}, \qquad \mbox{ for any } T \subseteq N \setminus S,
    \end{equation*}
    then
     \begin{align*}
         C_{SO}^v(S)=C_{pq}^v(S).
    \end{align*}

\end{proposition}

\begin{proof}
Observe that, for all $\sigma\in\Sigma_S(N)$ such that $Pre^\sigma(S)=T$ and the elements in $S$ arrive in the same order, the value of $AS^v(S,\sigma)$ is constant. Moreover there are $r!(s-r)!$ sequences $\sigma$ of the players divided in the two sets of \(\pi\) in the specified order $\{R,S\setminus R\}$ such that $BI^v(\{F_r^{\sigma}(S),L_r^{\sigma}(S)\},Pre^\sigma(S))$ is constant. Then
\begin{equation*}
  \mathcal{A}^v_{SO}(S,T) = \frac{r!(s-r)!}{s!(s-1)}\sum_{r=1}^{s-1}\sum_{ R\subset S: |R|=r}BI^v(\{R,S\setminus R\},T)=\frac{2r!(s-r)!}{s!(s-1)}\sum_{\pi\in\Pi_2(S)}BI^v(\pi,T). 
\end{equation*}
To conclude the proof note that there are \(t!(n-s-t)!r!(s-r)!\) sequences of the players in which all the players of \(T\) are followed by the two sets of \(\pi\) in a specified order and followed by the remaining players.
If we denote as \(\Sigma_{T,\pi}\) all the sequences in which \(T\) is followed by \(\pi\) in any of the two orders, the Shapley Owen value of \(S\) relative to these sequences is
\[
C^v_{SO}(S;T,\pi):=\frac{2 t!r!(s-r)!(n-s-t)!}{(n-s+1)!s!}\sum_{\sigma \in \Sigma_{T,\pi}} \frac{1}{s-1} \sum_{r=1}^{s-1} BI^v(\{F_r^{\sigma}(S),L_r^{\sigma}(S)\},Pre^\sigma(S)).
\]

Summing up over all sequences in \(\Sigma_S(N)\), that corresponds to sum up over all $T\subseteq N \setminus S$ and $\pi\in\Pi_2(S)$, we have
\begin{align*}
C^v_{SO}(S) &= \sum_{T \subseteq N \setminus S} \sum_{\pi \in \Pi_2(S)} C^v_{SO}(S;T,\pi) \\
& =\sum_{T \subseteq N \setminus S} \frac{t!(n-s-t)!}{(n-s+1)!} \sum_{\{R_1,R_2\} \in \Pi_2(S)} \frac{2  r!(s-r)!}{s!(s-1)}  BI^v(R_1,R_2,T)
\end{align*}
which yields the result.
     \end{proof}

 The  Shapley-Owen absolute coopetition index reads as

 \begin{equation*}
     \hat{\mathcal{C}}^v_{SO}(S)=\frac{\sum_{T \subseteq N \setminus S} t!(n-s-t)! \sum_{\{R_1,R_2\} \in \Pi_2(S)} \frac{2  r!(s-r)!}{s!(s-1)}  BI^v(R_1,R_2,T)}{\sum_{T \subseteq N \setminus S} t!(n-s-t)!\,v'_S(T)}.
 \end{equation*}

\subsection{Some Common Properties}
As already noted, the attitude for couples does not depend on the distribution \(p\), and, therefore,  \(\mathcal{A}^v_{Bz}(ij,T)=\mathcal{A}^v_{SO}(ij,T) = c^v(ij,T)\) for any couple \(ij \subset S\) and any \(T \subseteq N \setminus ij\). Moreover, \( \mathcal{C}^v_{SO}(ij) = \mathcal{I}_{Sh}^v(ij)\). The next result shows that the equality between the two indices extends to the triplets of players

\begin{proposition}
    Take $i,j,k\in N$. Then, for every $T \subset N\setminus ijk$, the Banzhaf and Shapley attitude indices coincide:
    \begin{align*}
        \mathcal{A}^v_{Bz}(ijk,T) = \mathcal{A}^v_{SO}(ijk,T).
    \end{align*}
\end{proposition}
\begin{proof}
Using the definition we have
\begin{align*}
\mathcal{A}^v_{SO}(ijk,T) & = \frac{1}{12}\Big[2BI^v(\{i, jk\},T) + 2BI^v(\{j, ik\}, T) + 2 BI^v(\{k, ij\},T) \\
&\qquad + 2BI^v(\{jk, i\},T) + 2BI^v(\{ik, j\}, T) + 2BI^v(\{ij, k\},T)  \Big]\\
& = \frac{1}{3}\Big[BI^v(\{i, jk\}, T) + BI^v(\{j, ik\}, T) BI^v(\{k, ij\}, T) \Big] = 
\mathcal{A}^v_{Bz}(ijk,T). 
\end{align*}    
\end{proof}

We apply the new indices to the introductory examples.
\begin{example}[continues=ex:profitparadox]
    Considering again the game $v$, we have $\mathcal{C}_{Bz}^v(34)=-\frac{1}{4}$,  $\hat{\mathcal{C}}_{Bz}^v(34)=-\frac{2}{5}$  and $\mathcal{C}_{SO}^v(34)=\mathcal{I}^v(34)=-\frac{1}{3}$, $\hat{\mathcal{C}}_{SO}^v(34)=-\frac{4}{7}$. These indices show a troubled coexistence between 3 and 4 which was not captured by the profitability index.
    \end{example}
    
    \begin{example}[continues=ex:interparadox]
    Considering again  the games $v$, $w$, and $w$ , we have
    $\mathcal{C}_{Bz}^v(123)=\hat{\mathcal{C}}_{Bz}^v(123)=1$, $\mathcal{C}_{Bz}^w(123)=\hat{\mathcal{C}}_{Bz}^w(123)=0$, $\mathcal{C}_{Bz}^u(123)=\hat{\mathcal{C}}_{Bz}^u(123)=-1$ and  $\mathcal{C}_{SO}^v(123)=\hat{\mathcal{C}}_{SO}^v(123)=1$, $\mathcal{C}_{SO}^w(123)=\hat{\mathcal{C}}_{SO}^w(123)=0$, $\mathcal{C}_{SO}^u(123)=\hat{\mathcal{C}}_{SO}^u(123)=-1$. In these games, there is either complete cooperation, perfect balance, or intense competition. 
\end{example}
        
These simple examples highlight how the new indices can help reveal the average attitude towards cooperation and competition of any coalition of size larger than 1 in a way that could not be measured with the existing tools.

\section{Index properties}\label{sec:property}

In this section we analyze some properties satisfied by the coopetitive index. We start observing that if $S$ is formed by null players then, for all families $p$ and $q$,

\[
\mathcal{C}^v_{p,q}(S)=0.
\]

Moreover,  including a null player to a coalition can only water down the power of its players to cooperate or to offer credible alternatives within the coalition. This loss of power is formally described by a decrease in the absolute value of the coopetition index.

\begin{proposition}\label{prop:Dnull}
    Let $i \in N$ a null player and $S\subseteq N\setminus i$. Take a pair of families of probability distributions $q$ and $p$  such that $p_S, p_{S\cup  i}, q_S, q_{S\cup  i}$ satisfy
    \begin{enumerate}
        \item[i.)]   $\dfrac{p_{S\cup  i}(\{S_1\cup i,S_2\}) + p_{S\cup  i}(\{S_1,S_2\cup i\})}{p_{S}(\{S_1,S_2\})} = K_p(S) \leq 1$,\quad $\forall\{S_1,S_2\}\in\Pi_2(S)$; 
        \item[ii.)]   $\dfrac{q_{S\cup i}(T)}{q_{S}(T) + q_{S}(T\cup i)} = K_q(S) \leq 1$,\quad $\forall T\in 2^{N\setminus (S\cup  i)}$.
    \end{enumerate}
    Then
    \begin{equation*}
        |\mathcal{C}^v_{p,q}(S\cup i)| \leq K_p(S)K_q(S) |\mathcal{C}^v_{p,q}(S)|.
    \end{equation*}
\end{proposition}

\begin{proof}
By definition
\begin{equation}\label{eq:CI_S}
    \begin{aligned}
    \mathcal{C}^v_{p,q}(S) & = \sum_{T\subseteq N\setminus (S\cup i)} q_S(T)\mathcal{A}^v_{p}(S,T) + \sum_{ i\subset T\subseteq N\setminus S} q_S(T)\mathcal{A}^v_{p}(S,T) \\
    & = \sum_{T\subseteq N\setminus (S\cup i)} q_S(T)\mathcal{A}^v_{p}(S,T) + \sum_{ T\subseteq N\setminus (S\cup i)} q_S(T\cup i)\mathcal{A}^v_{p}(S,T\cup i)\\
    & = \sum_{T\subseteq N\setminus (S\cup i)} q_S(T)\mathcal{A}^v_{p}(S,T) + \sum_{ T\subseteq N\setminus (S\cup i)} q_S(T\cup i)\mathcal{A}^v_{p}(S,T)\\
    & = \sum_{T\subseteq N\setminus (S\cup i)} \Big(q_S(T)+q_S(T\cup i)\Big)\mathcal{A}^v_{p}(S,T),
\end{aligned}
\end{equation}
where the third equality follows from the hypothesis that $ i$ is a null player. Observe that $\Pi(S\cup i)=\big\{\{S_1\cup  i,S_2\},\{S_1,S_2\cup  i\}: \{S_1,S_2\}\in \Pi_2(S)\}\cup\{ i,S\}$, then the attitude of $S\cup i$ towards $T$ can be written
\begin{equation}\label{eq:BI_Se}
    \begin{aligned}
    \mathcal{A}^v_{p}(&S\cup i,T) =  \sum_{\{S_1,S_2\}\in\Pi_2(S)}p_{S\cup i}(\{S_1\cup i,S_2\})BI^v(\{S_1\cup i,S_2\},T) \\
    &\qquad\qquad\qquad+ p_{S\cup i}(\{S_1,S_2\cup i\})BI^v(\{S_1,S_2\cup i\},T)  \\
    &\qquad\qquad\qquad+ p_{S\cup i}(\{ i,S\})BI^v(\{ i,S\},T) \\
    = & \sum_{\{S_1,S_2\}\in\Pi_2(S)}p_{S\cup i}(\{S_1\cup i,S_2\})BI^v(\{S_1,S_2\},T) + p_{S\cup i}(\{S_1,S_2\cup i\})BI^v(\{S_1,S_2\},T)\\
    = & \sum_{\{S_1,S_2\}\in\Pi_2(S)}\Big(p_{S\cup i}(\{S_1\cup i,S_2\})+ p_{S\cup i}(\{S_1,S_2\cup i\})\Big)BI^v(\{S_1,S_2\},T),
\end{aligned}
\end{equation}
once more we use that $ i$ is a null player. Using \eqref{eq:BI_Se} and hypothesis $i)$ we have
\begin{equation*}
    \mathcal{A}^v_{p}(S\cup i,T) = K_p(S) \mathcal{A}^v_{p}(S,T).
\end{equation*}
Then, by \eqref{eq:CI_S},
\begin{equation*}
    \mathcal{C}^v_{p,q}(S\cup i) = K_q(S)K_p(S) \mathcal{C}^v_{p,q}(S),
\end{equation*}
and the proof is complete.
\end{proof}

    The probability distributions underlying the  Banzhaf and the Shapley-Owen indices satisfy the hypothesis of Proposition \ref{prop:Dnull}. We therefore have the following
\begin{corollary}
     Let $ i\in N$ a null player and $S\subseteq N\setminus i$. Then
       \begin{gather*}
        |\mathcal{C}^v_{Bz}(S\cup i)| \leq \left( 1 - \frac{2}{2^s - 1} \right) |\mathcal{C}^v_{Bz}(S)|,
        \\
        |\mathcal{C}^v_{SO}(S\cup i)| \leq \left( 1 - \frac{2}{s(s+1)} \right) |\mathcal{C}^v_{SO}(S)|.
    \end{gather*}
\end{corollary}
\begin{proof}    
    For Banzhaf, take $q_S(T)=\frac{1}{2^{n-s}}$ and $p_S(\{S_1,S_2\})=\frac{1}{2^{s-1}-1}$, then we have
    \begin{align*}
  &\frac{p_{S\cup  i}(\{S_1\cup i,S_2\}) + p_{S\cup  i}(\{S_1,S_2\cup i\})}{p_{S}(\{S_1,S_2\})} = \frac{\frac{1}{2^{s}-1} + \frac{1}{2^{s}-1}}{\frac{1}{2^{s-1}-1}} = 1-\frac{2}{2^{s}-1}<1;\\
  &\frac{q_{S\cup i}(T)}{q_{S}(T) + q_{S}(T\cup i)} = \frac{\frac{1}{2^{n-s-1}}}{\frac{1}{2^{n-s}}+\frac{1}{2^{n-s}}} = 1.
    \end{align*}
    For Shapley-Owen we take $q_S(T)=\frac{(n-t-s)!t!}{(n-s+1)!}$ and $p_S(\{S_1,S_2\})=\frac{2s_1!s_2!}{(s-1)s!}$, then
    \begin{align*}
  &\frac{p_{S\cup  i}(\{S_1\cup i,S_2\}) + p_{S\cup  i}(\{S_1,S_2\cup i\})}{p_{S}(\{S_1,S_2\})} = \frac{\frac{(s_1+1)!s_2!}{s(s+1)!} + \frac{s_1!(s_2+1)!}{s(s+1)!}}{\frac{s_1!s_2!}{(s-1)s!}}  = 1-\frac{2}{s(s+1)}<1;\\ 
  &\frac{q_{S\cup i}(T)}{q_{S}(T) + q_{S}(T\cup i)} = \frac{\frac{(n-t-s-1)!t!}{(n-s)!}}{\frac{(n-t-s)!t!}{(n-s+1)!}+\frac{(n-t-s-1)!(t+1)!}{(n-s+1)!}} = 1.
  \end{align*}
\end{proof}

Null players have a different impact on the indices that we reviewed in the introduction. These can be explained by the axioms that these indices satisfy. The inclusion of a null player does not change the profitability of a coalition for the role of null players in the Shapley individual and group values. Conversely, null players drop (or raise) the interaction index to zero, due to the fact that such players will not make any difference in the interaction index of the remaining players.

Next, we analize the behavior of the coopetition index on the minimal winning coalitions. A secession in a minimal winning coalition can only bring about a more competitive attitude among the players of that coalition and its supersets. This results in a decrease of  the coopetition index of all these coalitions. 

\begin{proposition}
    Let $v$ be a simple game and $\wcm^v$ its minimal winning coalitions set. Take $W\in\wcm^v$, with $|W|\geq2$ and let $w$ and $u$ be simple games whose minimal winning coalitions sets are defined as 
    \begin{align*}
        &\wcm^w=\wcm^v\setminus W \cup \{ W_1,W_2\}\\
        &\wcm^w=\wcm^v\setminus W \cup \{ \hat{W}_1,\hat{W}_2\}
    \end{align*}
    where $W_1\cap W_2\neq\emptyset$, $W_1 \cup W_2=W$ and $\{\hat{W}_1,\hat{W}_2\}\in\Pi_2(W)$. Then, for any pair of families of probability distribution $p$ and $q$, 
\begin{equation}
    \label{eq:wcmbup1st_conclusion}
    \mathcal{C}^v_{p,q}(S) \geq \mathcal{C}^w_{p,q}(S) \geq \mathcal{C}^u_{p,q}(S), \qquad \mbox{for any } S \supseteq W.
\end{equation}
\end{proposition}

\begin{proof}
 Fix any $S \supseteq W$ and $T \subseteq N \setminus S $. Since $W \subseteq S \cup T$, then $v(S \cup T)=1$. Consider now $\pi=\{S_1,S_2\} \in \Pi_2(S)$. For the first inequality in \eqref{eq:wcmbup1st_conclusion} we have that  $BI^v(\pi,T) \geq BI^w(\pi,T)$ whenever any set of $\pi$ say $S_1$, contains either $W_1$ or $W_2$, but $S_2$ does not. Strict majorization takes place if $v(S_1 \cup T)=0$, while $w(S_1 \cup T)=1$.
Conversely, $BI^v(\pi,T) = BI^w(\pi,T)$ whenever either set of $\pi$ entirely contains $W$.
For the second inequality we have that
\begin{equation}
    \label{eq:mcwbup1st_proof}
    BI^w(\pi,T) \geq BI^u(\pi,T),
\end{equation}
whenever either set of $\pi$ contains $\hat{W}_1$, $\hat{W}_2$, respectively, without containing $W_1$, $W_2$, respectively. For all the other cases we have $BI^w(\pi,T) = BI^u(\pi,T)$.
\end{proof}


\section{Notable examples}\label{sec:examples}
We apply the newly defined indices to two classes of games that were considered by Hart and Kurz \cite{hart1984stable} to examine the stability of coalitions\footnote{The first draft of this work available on Arxiv.org, introduced a \em{decisiveness index} to distinguish between powerless coalitions and those in which competition and cooperation balance out. Our later findings show that the probabilistic generalized value does the same job in a neater way.}.

\subsection{The Apex Game}

The \emph{Apex game}  $(N,x)$ is defined as follow 
    \begin{equation*}
        x(S) = \begin{cases}
            1 \mbox{ if } a\in S  \mbox{ and } S\backslash a \neq \emptyset;\\
            1 \mbox{ if } S=N\backslash a;\\
            0 \mbox{ otherwise};
        \end{cases}
    \end{equation*}
 for an "apex" player $a\in N$. This game was first described by \cite{von1947theory} and studied by \cite{hart1984stable}. It is an example where both exclusive and inclusive contracts by a group of players are profitable according to the definition in \cite{segal2003collusion}.\\
 
 We now compute the coopetition index for each coalition.
 \begin{proposition}\label{prop:apex}
     Consider the Apex game $(N,x)$, for any $S\subseteq N$, \(s \geq 2\),  take a pair of families of probability distribution $p$ and $q$. The coopetition index takes on the following values:
   \begin{equation*}
    \mathcal{C}^x_{q,p}(S)=\begin{cases}
        q_S((S \cup a)^c) - q_S(a) & \mbox{if } a \notin S;
        \\
       \left( q_s(\emptyset) - q_S(S^c) \right)  p_S(\{a,S \setminus a\}) & \mbox{if } a\in S, S\neq N;
        \\
        0 & \mbox{if }S=N.
    \end{cases}
\end{equation*}
 \end{proposition}
 \begin{proof}
Suppose \(a \notin S\). For any \(\{S_1,S_2\} \in \Pi_2(S)\), we have
\begin{equation*}
      BI^x(\{S_1,S_2\},T)=\begin{cases}
          +1 \qquad \mbox{ if }  T =  (S \cup a)^c;\\
          -1 \qquad \mbox{ if }  T =  a;\\
          0 \qquad \mbox{ otherwise}.
      \end{cases}
  \end{equation*}
Therefore, by its definition,
  \begin{align*}
      \mathcal{C}_{q,p}^x(S)= & \sum_{T\subseteq N\setminus(a\cup S) } \left( q_S(T) \sum_{\{S_1,S_2\}\in\Pi_2(S)}p_S(\{S_1,S_2\}) BI^x(\{S_1,S_2\},T)  \right.\\
      & \quad+ \left.  q_S(T\cup a)\sum_{\{S_1,S_2\}\in\Pi_2(S)}p_S(\{S_1,S_2\}) BI^x(\{S_1,S_2\},T\cup a) \right) \\
      = & \left(q_S((S \cup a)^c) - q_S(a) \right)  \sum_{\{S_1,S_2\}\in\Pi_2(S)}p_S(\{S_1,S_2\})   \\ 
      = &  q_S((S \cup a)^c) - q_S(a).
  \end{align*}

  Consider now \(a \subset S \subset N\). We have that:
  \begin{equation*}
      BI^x(\{a,S \setminus a\},T)=\begin{cases}
          +1 \qquad \mbox{ if }  T =  \emptyset ;\\
          -1 \qquad \mbox{ if }  T =  S^c;\\
          0 \qquad \mbox{ otherwise },
      \end{cases}
  \end{equation*}
while
\[
BI^x(\{S_1\cup a,S_2\},T)=0, \quad \mbox{for any }
\{S_1,S_2\} \in \Pi_2(S) \setminus \{a,S \setminus a\}.
\]
 Therefore, considering the definition of the coopetitive index,
\begin{align*}
    \mathcal{C}^x_{p,q}(S) = & \sum_{T \subseteq N \setminus S}q_S(T) \mathcal{A}^x_p(S,T)  
    \\
    & = \big( q_s(\emptyset) - q_S(S^c) \big) p_S(\{a,S \setminus a\}).
\end{align*}

The case \( \mathcal{C}^x_{p,q}(N)=0\) simply follows  from
\[
BI^x(\{S_1,S_2\},\emptyset)=0, \quad \mbox{for any } \{S_1,S_2\} \in \Pi_2(S) .
\]

 \end{proof}

 Proposition \ref{prop:apex} implies that if, for any $S$ such that $2\leq s\leq n-1$,  
\begin{equation}
\label{eq:condnullapex}
\begin{split}
q_S\big(N\setminus( a\cup S) \big)=q_S\big(a),\qquad\text{for } a\notin S;\\
q_S\big(N\setminus S \big)=q_S\big(\emptyset),\qquad\text{for } a \in S,  
\end{split}
\end{equation}
  then $\mathcal{C}_{q,p}^x(S)=0$. 
  The distribution families that define the Banzhaf and the Shapley-Owen values satisfy condition \eqref{eq:condnullapex}. Therefore, for any \(S\), \(s \geq 2\), we have,
    \begin{equation*}
        \mathcal{C}^x_{Bz}(S)= \hat{\mathcal{C}}^x_{Bz}(S)=\mathcal{C}^x_{SO}(S)=\hat{\mathcal{C}}^x_{SO}(S)=0.
    \end{equation*}
According to all our indices, each coalition in the apex game shows a perfect balance between cooperation and competition as measured by our proposals. Note that the null values do not derive from a lack of power, since every coaliton has a positive generalized value according to the most popular indices. For instance, the group Shapley value  takes the following values that are strictly positive for any coalition:
\begin{align*}
    &\phi_{Sh}(S)=\frac{2}{(n-s)(n-s-1)},\ \forall S\subset N\setminus a \text{ with }s\geq 2;\\
    &\phi_{Sh}(S\cup a)=1,\ \forall S\subseteq N\setminus a \text{ with }s\geq 2;\\
    &\phi_{Sh}(N\setminus a)=\frac{1}{2}.
\end{align*}

\subsection{Symmetric Majority Games}
The next application focuses on the study of $n$-person symmetric majority games, $(N,m^k)$, defined as follows:
\[ m^k(S) =
\begin{cases}
1 & \text{if } s \ge k, \\
0 & \text{otherwise},
\end{cases} \]
where \( k \) is the quota majority required to win. It is typically assumed that \( k \) is an integer greater than \( n/2 \) (ensuring the game is super-additive). These games were initially examined by Bott \cite{bott1953majority}, who characterized all their symmetric solutions. \\

In the following proposition we give the explicit expression of the coopetitive index.
\begin{proposition}\label{prop:majority}
     Consider the symmetric majority game $(N,m^k)$, where $k$ is the quota. For any $S\subseteq N$, \(s \geq 2\),  take any pair of families of probability distributions $p$ and $q$. The coopetition index takes on the following values:
     \begin{itemize}
         \item If $s\leq k-1$ and $s$ is even then
   \begin{align*}\label{eq:MajoritypqCoptA}
    \mathcal{C}^m_{q,p}(S)=\sum_{t=k-s}^{\min\{k-1,n-s\}}\sum_{\substack{T\subseteq N\setminus S:\\|T|=t}}q_S(T)\Big( \mathds{1}_{t\leq k-\sfrac{s}{2}-1}&\sum_{a=t-k+s+1}^{\sfrac{s}{2}}\sum_{\substack{ \{S_1,S_2\}\in\Pi_2:\\|S_1|=a}}p_S(S_1,S_2)\\
    &- \mathds{1}_{t\geq k-\sfrac{s}{2}}\sum_{a=k-t}^{\sfrac{s}{2}}\sum_{\substack{ \{S_1,S_2\}\in\Pi_2:\\|S_1|=a}}p_S(S_1,S_2)\Big).
\end{align*}
\item If $s\leq k-1$ and $s$ is odd 
then
\begin{align*}
    \mathcal{C}^m_{q,p}(S)=\sum_{t=k-s}^{\min\{k-1,n-s\}}\sum_{\substack{T\subseteq N\setminus S:\\|T|=t}}q_S(T)\Big( \mathds{1}_{t\leq k-\lceil\sfrac{s}{2}\rceil-2}&\sum_{a=t-k+s+1}^{\lceil\sfrac{s}{2}\rceil}\sum_{\substack{ \{S_1,S_2\}\in\Pi_2:\\|S_1|=a}}p_S(S_1,S_2)\\
    &- \mathds{1}_{t\geq k-\lceil\sfrac{s}{2}\rceil}\sum_{a=k-t}^{\lceil\sfrac{s}{2}\rceil}\sum_{\substack{ \{S_1,S_2\}\in\Pi_2:\\|S_1|=a}}p_S(S_1,S_2)\Big).
\end{align*}
\item If $s=k+b$ with $0\leq b\leq n-k$ then
\begin{align*}
    \mathcal{C}^m_{q,p}(S)=\sum_{t=0}^{n-s}\sum_{\substack{T\subseteq N\setminus S:\\|T|=t}}q_S(T)\Big( \mathds{1}_{t<\frac{k-b}{2}}&\sum_{a=\min\{t+b,\lceil\sfrac{s}{2}\rceil\}+1}^{\min\{k-t-1,\lceil\sfrac{s}{2}\rceil\}}\sum_{\substack{ \{S_1,S_2\}\in\Pi_2:\\|S_1|=a}}p_S(S_1,S_2)\\
    &- \mathds{1}_{t\geq \frac{k-b}{2}}\sum_{a=\min\{k-t,\lceil\sfrac{s}{2}\rceil\}}^{\min\{t+b,\lceil\sfrac{s}{2}\rceil\}}\sum_{\substack{ \{S_1,S_2\}\in\Pi_2:\\|S_1|=a}}p_S(S_1,S_2)\Big).
\end{align*}
\end{itemize}
\end{proposition}

The proof is postponed to the appendix \ref{appendix1}. \\

For any $S\subseteq N$, with $s\geq 2$, the group Shapley value is given by
\begin{align*}
    \phi_{Sh}(S)=\begin{cases}
        \frac{1}{n-s-1}\min\{s,n-k+1\},\quad\text{ if } s\leq k-1,\\
        1, \quad\text{ otherwise. }
    \end{cases}
\end{align*}

Moreover, we have the following characterizion of the Shapley coopetitive index.

\begin{proposition}\label{prop:cbella}
The Shapley coopetitive index for the symmetric majority game \((N,m^k)\) has the form 
    \begin{itemize}
    \item If $s\leq k-1$ then
     \begin{equation}\label{eq:SOMajorityi)}
        \mathcal{C}_{SO}^m(S)=\frac{\max\big\{0,s-1+k-n\big\}}{s-1}\phi_{Sh}(S).
        \end{equation}

       \item If $s=k+b$ with $0\leq b\leq n-k$ then
       \begin{equation}\label{eq:SOMajorityii)}
        \mathcal{C}_{SO}^m(S)=\frac{2k-1-n}{(s-1)}.
        \end{equation}
\end{itemize}
\end{proposition}

The proof of this result is postponed to appendix \ref{appendix1}. \\

It is easy to verify that $0 \leq \frac{\max\{0,s-1+k-n\}}{s-1} <1$, for all $k,s,n$. Moreover, from  Proposition \ref{prop:cbella} we have the following corollary.\\

\begin{corollary}
The Shapley coopetitive index for the symmetric majority game \((N,m^k)\) of coalition $S$ is zero  whether $s\leq n-k+1$  or $n$ is odd, $k=\lceil\frac{n}{2}\rceil +1$, and $k\leq s\leq n$; otherwise, it is strictly positive. 
\end{corollary}
 

Therefore, the competitive effort never prevails over the cooperative one. In some instances, such as the case of small size coalitions or the case where there is a odd number of voters and the quota is minimal, independently of the size, the two efforts balance evenly.

\section{An Electoral Application}\label{sec:elect}

We consider a voting game based on actual electoral data, previously examined as Example 1.3 by Flores et al. \cite{flores2019shapley}. This case refers to the elections for the Spanish Congress of Deputies held on December 20, 2015, to elect the 11th Cortes Generales. The resulting distribution of seats, as shown in Table \ref{tab:congress_seats}, defines a weighted voting game with a quota of \( k=176 \). 

\begin{table}[h]
    \centering
    \begin{tabular}{|l|c|}
        \hline
        \textbf{Party} & \textbf{Seats} \\
        \hline
        PP & 123 \\
        PSOE & 90 \\
        Podemos & 69 \\
        Ciudadanos & 40 \\
        ERC & 9 \\
        DyL & 8 \\
        PNV & 6 \\
        Unidad Popular (UP) & 2 \\
        Bildu & 2 \\
        Coalici\'on Canaria & 1 \\
        \hline
    \end{tabular}
    \caption{Congress of Deputies in the 11th Spanish Cortes Generales.}
    \label{tab:congress_seats}
\end{table}

This election resulted in a fragmented Parliament, leading to prolonged negotiations to form a stable government. These consultations were heavily influenced by the political affinity among parties, as coalitions between ideologically divergent parties could only materialize if intermediary parties facilitated reconciliation. These constraints on coalition formation can be effectively represented by a communication graph, where a coalition retains its value only if its members are connected. The graph used in our analysis, taken from Flores et al. \cite{flores2019shapley}, is shown in Figure \ref{fig:political_affinities}.

\begin{figure}[h]
    \centering
    \includegraphics[width=0.8\textwidth]{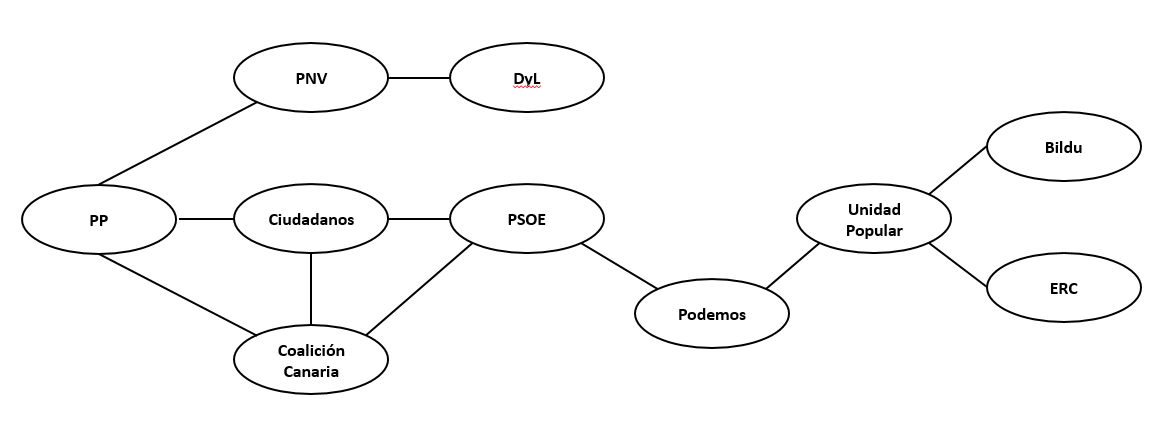}
    \caption{Graph of political affinities among the parties after the 2015 election.}
    \label{fig:political_affinities}
\end{figure}

In this model, a coalition is only valuable if its members form a connected subgraph of the original graph. The coalition's worth is determined by the Myerson Value, which adjusts the game to consider only the sum of the values of its connected components. The Myerson graph-restricted game \(v_{\Gamma}\) modifies the original voting game \(v\) as follows:
\[
    v_{\Gamma}(S) = \sum_{T_k \in \mathrm{con}_{\Gamma}(S)} v(T_k),
\]
where \(\mathrm{con}_{\Gamma}(S)\) represents the set of connected components of \(S\) in \(\Gamma\). Since the voting game is proper and the connected components are disjoint, only one of them at most can be a winning coalition. 

For each coalition, we compute both the Shapley-Owen coopetition and absolute coopetition indices, and we analyze their distributions. Figures \ref{fig:coopetition_boxplots} and \ref{fig:absolute_coopetition_boxplots} display boxplots showing the coopetition index as a function of coalition size and of the number of connected components.

\begin{figure}[h]
    \centering
    \includegraphics[width=0.48\textwidth]{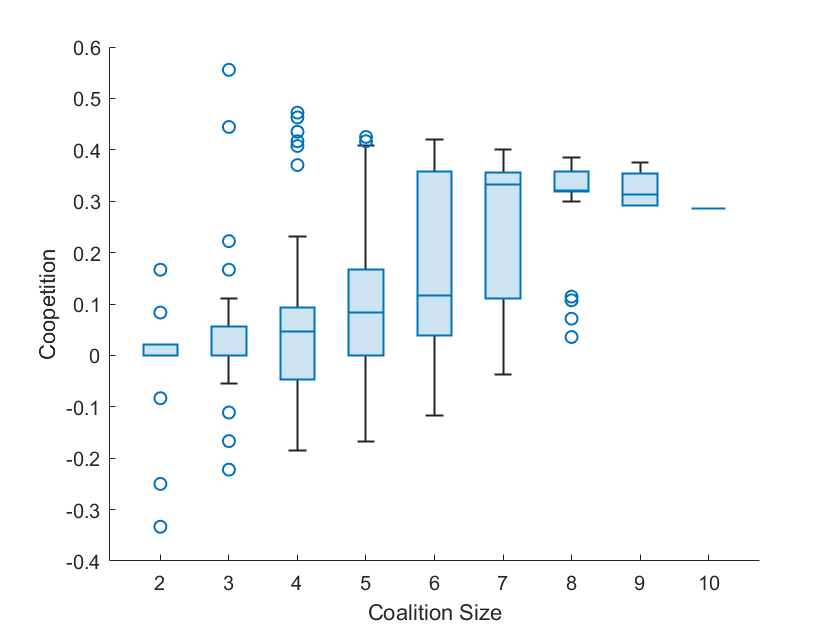}
    \includegraphics[width=0.48\textwidth]{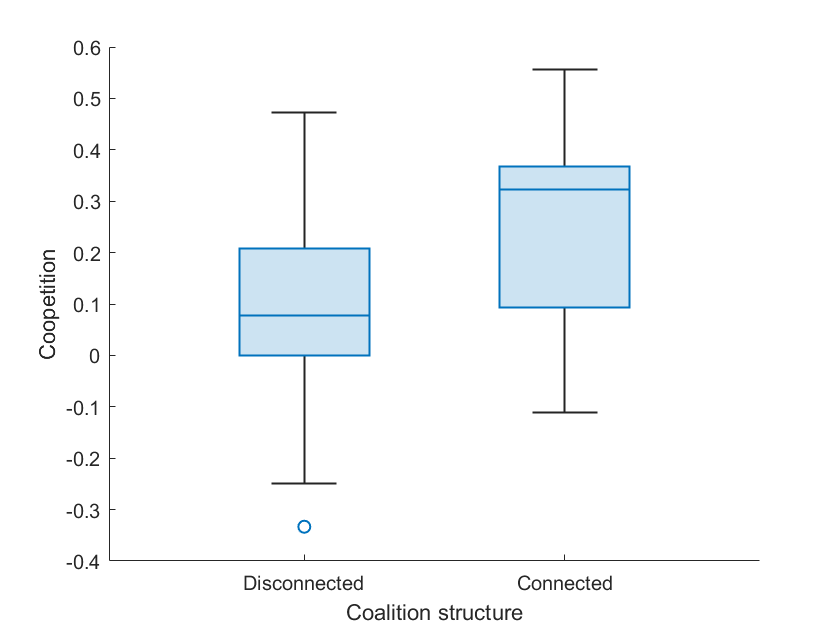}
    \caption{Boxplots of the Shapley-Owen coopetition index arranged by coalition length (left) and connectedness (right).}
    \label{fig:coopetition_boxplots}

 \end{figure}
 
\begin{figure}[h]
    \centering
    \includegraphics[width=0.48\textwidth]{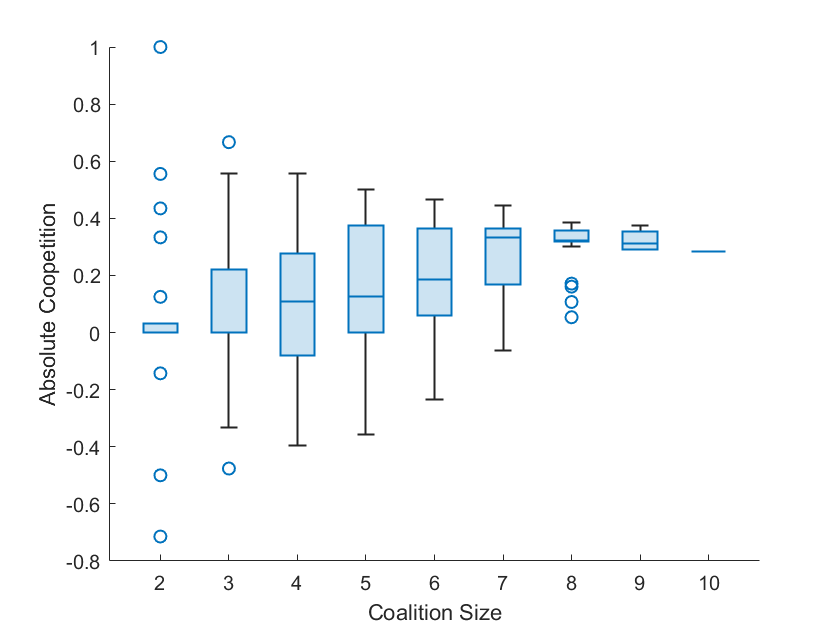}
    \includegraphics[width=0.48\textwidth]{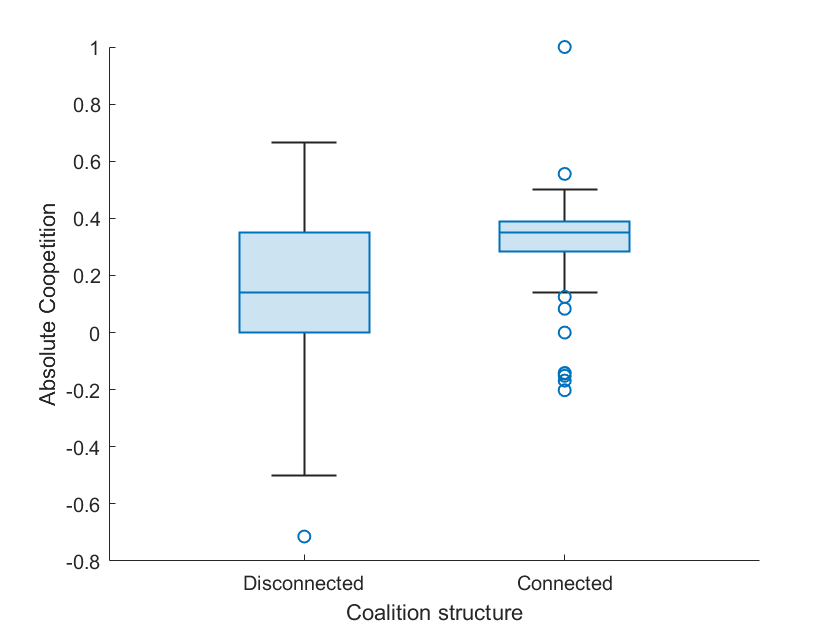}
    \caption{Boxplots of the absolute Shapley-Owen coopetition index arranged by coalition length (left) and connectedness (right).}
    \label{fig:absolute_coopetition_boxplots}
\end{figure}

The Wilcoxon Rank-Sum tests reveal a statistically significant increase in the population median when coalitions are fully connected, with a null $p$-value up to the tenth decimal digit.

Flores et al. \cite{flores2019shapley} specifically analyzed two potential coalitions during this period: one between PSOE and Ciudadanos, an emerging right-leaning party, and another between PSOE, Podemos, and Unidad Popular. Their study concluded that the PSOE-Ciudadanos coalition was more promising in terms of group value and profitability. However, when considering coopetition, this coalition had a lower coopetition and absolute coopetition index than the PSOE-Podemos-UP coalition. Specifically, if \(S' = \{ \mbox{PSOE},\mbox{Ciudadanos} \}  \) and \(S'' = \{ \mbox{PSOE},\mbox{Podemos},\mbox{UP} \}  \), then:
\begin{gather*}
     \mathcal{C}^v_{SO}(S')= 0.0833,  \quad \hat{\mathcal{C}}^v_{SO}(S')= 0.125, \\
    \mathcal{C}^v_{SO}(S'')= 0.111,  \quad \hat{\mathcal{C}}^v_{SO}(S'')= 0.289.
\end{gather*}

It is worth noting that this example was analyzed while negotiations to form a government were still ongoing. The talks that followed  ultimately stalled, prompting new elections. Consequently, none of the computed indices could be tested against a successful outcome.

\section{Conclusion and Open Problems}\label{sec:conclusion}

The coopetition indices defined here provide a novel way for evaluating the attitudes towards cooperation and competition for coalitions of any size larger than 1 in simple games. The new measures overcome some difficulties arising from borrowing the existing definitions arising from the fact that those indices are defined with different aims or from the fact that the available tools alter the original game. 

We suggest accompanying the coopetition index with the corresponding generalized group value. Together, the two measures reveal the strength and the cohesion of the coalition under exam. This is particularly important in situations such as the apex game, where cooperative and competitive forces perfectly balance. Additionally, the indices' ratio introduces an absolute coopetition index with attainable range ends of 1 (for maximum cooperation) and -1 (for maximum competition) for any coalition.

There are many open questions and we list what we consider the three most urgent ones:
\begin{enumerate}
    \item While we have proved some properties for the new indices, we have not yet provided a list of independent axioms that uniquely characterize the indices for some choices of the distribution families \(p\) and \(q\). The Banzhaf and the Shapley-Owen coopetition indices are the first targets.
    \item The indices have been defined for simple games because this is the clearest environment to explain the block interaction indicator and the attitude index. The indices lend themselves to be effortlessly adapted to Transferrable Utility (TU) games. The meaning and reasonability of the new definitions have to be carefully tested in the larger environment.
    \item One of the referees has pointed out that the attitude index captures only the 2-partitions of a coalition and does not take into account partitions with a higher number of non-empty disjoint sets. We reply by noting that our attitude index will assign a negative value to a 3-partition of complementary critical sets lower than that assigned to a 2-partition of complementary critical disjoint subsets of the same coalition. Nevertheless, we fully accept the criticism and leave as an open problem the idea of defining attitude (and coopetition) indices based on splits of any cardinality.
\end{enumerate}

Future research addressing these questions will further refine the theoretical foundations of the coopetition index and enhance its practical applicability in evaluating coalition dynamics.\\

 \vspace{1cm}
\noindent\textbf{\emph{Compliance with Ethical Standard}}\\
 This article does not contain any studies with human participants or animals performed by any of the authors.
 \vspace{0.4cm}\\
\noindent\textbf{\emph{Acknowledgements}}\\
The authors gratefully acknowledge the anonymous reviewers for their useful comments and suggestions that allowed the work to improve.
\nocite{}
\bibliographystyle{abbrv}
\bibliography{references}

\appendix
\section{Appendix: Coopetitive index in the majority game}\label{appendix1}
In this section we start giving the proof of Proposition \ref{prop:majority}.
\begin{proof}
Take a coalition $S\subseteq N$ with cardinality $s$ and let $T\subseteq N\setminus S$ be a loosing coalition with cardinality  $t$. Coalition $S$ is critical for $T$ if and only if $s\geq k-t$.  We need to distinguish three cases: $(A)$ Coalition $S$ is losing and $s$ is even; $(B)$ Coalition $S$ is losing and $s$ is odd; $(C)$ Coalition $S$ is winning. For each case we write explicitly the value of the block interaction indicator for any $T$.\\
{\bf Case $(A)$: }
As already observed we need that the cardinality of $T$ is $t=k-s + b$ for $b=0,\ldots,\min\{s-1,n-k\}$. Then if 

\begin{align*}
&t=k-s \Rightarrow BI^v(\pi, T) = 1 \quad \forall \pi \in \Pi_2(S)\\  
&t=k-s+1 \Rightarrow BI^v(\pi, T) =
\begin{cases} 
0 & \text{if } \pi = \{i, S \setminus i\}, \, \forall i \in S \\
1 & \text{otherwise}
\end{cases} \\
&t=k-s+2 \Rightarrow BI^v(\pi, T) =
\begin{cases} 
0 & \text{if } \pi = \{A, S \setminus A\}, \, \forall A \subseteq S, a \leq 2 \\
1 & \text{otherwise}
\end{cases}\\
&\cdots\\
&t=k-s+\frac{s}{2}-1 \Rightarrow BI^v(\pi, T) =
\begin{cases} 
0 & \text{if } \pi = \{A, S \setminus A\}, \, \forall A \subseteq S, a \leq \frac{s}{2}-1 \\
1 & \text{otherwise}
\end{cases}\\
&t=k-s+\frac{s}{2} \Rightarrow BI^v(\pi, T) =
\begin{cases} 
-1 & \text{if } \pi = \{A, S \setminus A\}, \, \forall A \subseteq S, a = \frac{s}{2} \\
0 & \text{otherwise}
\end{cases}\\
&t=k-s+\frac{s}{2}+1 \Rightarrow BI^v(\pi, T) =
\begin{cases} 
-1 & \text{if } \pi = \{A, S \setminus A\}, \, \forall A \subseteq S, a \geq \frac{s}{2}-1 \\
0 & \text{otherwise}
\end{cases}
\end{align*}
\begin{align*}
&\cdots\\
&t=k-s+s-2 \Rightarrow BI^v(\pi, T) =
\begin{cases} 
-1 & \text{if } \pi = \{A, S \setminus A\}, \, \forall A \subseteq S, a \geq 2 \\
0 & \text{otherwise}
\end{cases}\\
&t=k-s+s-1 \Rightarrow BI^v(\pi, T) = -1 \quad \forall \pi \in \Pi_2(S)
\end{align*}

{\bf Case $(B)$: }
As before we consider the cardinality $t=k-s + b$ for $b=0,\ldots,\min\{s-1,n-k\}$. Then if

\begin{align*}
&t=k-s \Rightarrow BI^v(\pi, T) = 1 \quad \forall \pi \in \Pi_2(S)\\  
&t=k-s+1 \Rightarrow BI^v(\pi, T) =
\begin{cases} 
0 & \text{if } \pi = \{i, S \setminus i\}, \, \forall i \in S \\
1 & \text{otherwise}
\end{cases} \\
&t=k-s+2 \Rightarrow BI^v(\pi, T) =
\begin{cases} 
0 & \text{if } \pi = \{A, S \setminus A\}, \, \forall A \subseteq S, a \leq 2 \\
1 & \text{otherwise}
\end{cases}\\
&\cdots\\
&t=k-s+\lceil\frac{s}{2}\rceil-1 \Rightarrow BI^v(\pi, T) =
\begin{cases} 
0 & \text{if } \pi = \{A, S \setminus A\}, \, \forall A \subseteq S, a \leq \lceil\frac{s}{2}\rceil-1 \\
1 & \text{otherwise}
\end{cases}\\
&t=k-s+\lceil\frac{s}{2}\rceil \Rightarrow BI^v(\pi, T) = 0,\ \forall\pi\in\Pi_2(S)\\
&t=k-s+\lceil\frac{s}{2}\rceil+1 \Rightarrow BI^v(\pi, T) =
\begin{cases} 
-1 & \text{if } \pi = \{A, S \setminus A\}, \, \forall A \subseteq S, a = \lceil\frac{s}{2}\rceil \\
0 & \text{otherwise}
\end{cases}\\
&t=k-s+\lceil\frac{s}{2}\rceil+2 \Rightarrow BI^v(\pi, T) =
\begin{cases} 
-1 & \text{if } \pi = \{A, S \setminus A\}, \, \forall A \subseteq S, a \geq \lceil\frac{s}{2}\rceil-1 \\
0 & \text{otherwise}
\end{cases}\\
&\cdots\\
&t=k-s+s-2 \Rightarrow BI^v(\pi, T) =
\begin{cases} 
-1 & \text{if } \pi = \{A, S \setminus A\}, \, \forall A \subseteq S, a \geq 2 \\
0 & \text{otherwise}
\end{cases}\\
&t=k-s+s-1 \Rightarrow BI^v(\pi, T) = -1 \quad \forall \pi \in \Pi_2(S)
\end{align*}

{\bf Case $(C)$: }
Coalition $S$ is winning then we calculate the block interaction indicator when $s=k+b$ for $b=0,\ldots,n-k$. Take $A\subseteq S$, with cardinality $a\geq1$, and take a partition $\pi=\{A,N\setminus A\}\in\Pi_2(S)$. For any loosing coalition $T$ with cardinality $t\leq n-s$ observe that $BI^v(\pi, T)=-1$ if and only if 
$$
t+a\geq k \qquad \mbox{and} \qquad t+k+b-a \geq k
$$
The previous inequality are compatible if $k-t\leq t+b$, namely $t\geq \frac{k-b}{2}$. Then

\begin{align*} 
&t\geq \frac{k-b}{2} \Rightarrow BI^v(\pi, T) =
\begin{cases} 
-1 & \text{if } \pi = \{A, S \setminus A\}, \, \forall A \subseteq S, k-t\leq a\leq \min\{t+b,\lceil\sfrac{s}{2}\rceil\} \\
0 & \text{otherwise}
\end{cases} \\
&t< \frac{k-b}{2} \Rightarrow BI^v(\pi, T) =
\begin{cases} 
0 & \text{if } \pi = \{A, S \setminus A\}, \, \forall A \subseteq S,a\leq \min\{t+b,\lceil\sfrac{s}{2}\rceil\} \mbox{ or }  k-t\leq a \leq \lceil\sfrac{s}{2}\rceil \\
1 & \text{otherwise}
\end{cases}
\end{align*}
Using the previous calculations the formulas in the proposition hold.
\end{proof}
Now we prove the explicit expression of the Shapley coopetition index given in Proposition \ref{prop:cbella}.
\begin{proof}
    To get the previous expression, first observe that applying Proposition \ref{prop:majority} we have the following expression of the Shapley-Owen coopetitive index.
\begin{itemize}
    \item If $s\leq k-1$ and $s$ is even then
    \begin{align*}
        \mathcal{C}_{SO}^m(S)=\frac{1}{(s-1)(n-s+1)}\sum_{t=k-s}^{\min\{k-1,n-s\}}&\mathds{1}_{t\leq k-\sfrac{s}{2}-1}\big(2k-2t-s-1\big)\\
        &-\mathds{1}_{t\geq k-\sfrac{s}{2}}\big(s+2t-2k+1\big).
    \end{align*}
    \item If $s\leq k-1$ and $s$ is odd then
    \begin{align*}
        \mathcal{C}_{SO}^m(S)=\frac{2}{(s-1)(n-s+1)}\sum_{t=k-s}^{\min\{k-1,n-s\}}&\mathds{1}_{t\leq k-\lceil\sfrac{s}{2}\rceil-2}\big(k-t-\lceil\sfrac{s}{2}\rceil-1\big)\\
        &-\mathds{1}_{t\geq k-\lceil\sfrac{s}{2}\rceil}\big(\lceil\sfrac{s}{2}\rceil+t-k+1\big).
    \end{align*}
    \item If $s=k+b$ with $0\leq b\leq n-k$ then
\begin{align*}
    \mathcal{C}^m_{SO}(S)=\frac{1}{(s-1)(n-s+1)}\sum_{t=0}^{n-s}&\mathds{1}_{t< \sfrac{1}{2}(k-b)}\big(2\min\{k-t-1,\lceil\sfrac{s}{2}\rceil\}-2\min\{t+b,\lceil\sfrac{s}{2}\rceil\}\\
    &\qquad\qquad\quad-\mathds{1}_{\min\{k-t-1,\lceil\sfrac{s}{2}\rceil\}=\sfrac{s}{2}}\big)\\
        -&\mathds{1}_{t\geq\sfrac{1}{2}(k-b)}\big(2\min\{t+b,\lceil\sfrac{s}{2}\rceil\}-2\min\{k-t,\lceil\sfrac{s}{2}\rceil\}+2\\
    &\qquad\qquad\quad-\mathds{1}_{\min\{t+b,\lceil\sfrac{s}{2}\rceil\}=\sfrac{s}{2}}\big).
\end{align*}
\end{itemize}
Rewriting the previous sums we get formulas \ref{eq:SOMajorityi)} and \ref{eq:SOMajorityii)}.\\
\end{proof}

\end{document}